\def\year{2018}\relax
\documentclass[letterpaper]{article} %DO NOT CHANGE THIS
\usepackage{aaai18}  %Required
\usepackage{times}  %Required
\usepackage{helvet}  %Required
\usepackage{courier}  %Required
\usepackage{url}  %Required
\usepackage{graphicx}  %Required
\usepackage{tikz}
\usepackage{subfig}
\usepackage{float}
\usetikzlibrary{positioning,chains,fit,shapes,calc}

\usepackage{amsthm}
\usepackage{amsmath}
\newtheorem*{theorem}{Theorem}

\newenvironment{claim}[1]{\par\noindent\textbf{Claim.}\space#1}{}

\graphicspath{ {images/} }
\definecolor{myblue}{RGB}{80,80,160}
\definecolor{mygreen}{RGB}{80,160,80}
\begin{document}
% The file aaai.sty is the style file for AAAI Press 
% proceedings, working notes, and technical reports.

\title{AUPCR Maximizing Matchings : Towards a Pragmatic Notion of Optimality for One-Sided Preference Matchings} % How about this title? -> Rahul
% \title{ Experimental Analysis of Notions of Optimality for One-Sided Preference Matchings} % this title? -> Girish
\author{ Girish Raguvir J\thanks{All authors contributed equally},  Rahul Ramesh\footnotemark[1],  Sachin Sridhar\footnotemark[1], Vignesh Manoharan\footnotemark[1]\\
Department of Computer Science and Engineering\\
Indian Institute of Technology Madras, India
}
\maketitle
\begin{abstract}
We consider the problem of computing a matching in a bipartite graph in the presence of one-sided preferences.
There are several well studied notions of optimality which include pareto optimality, rank maximality, fairness and popularity. In this paper, we conduct an in-depth experimental study comparing different notions of optimality based on a variety of metrics like cardinality, number of rank-1 edges, popularity, to name a few. Observing certain shortcomings in the standard notions of optimality, we propose an algorithm which maximizes an alternative metric called the \textit{Area under Profile Curve ratio} (AUPCR). To the best of our knowledge, the AUPCR metric was used earlier but there is no known algorithm to compute an AUPCR maximizing matching. Finally, we illustrate the superiority of the AUPCR-maximizing matching by comparing its performance against other optimal matchings on synthetic instances modeling real-world data.
\end{abstract}
\section{Introduction}
The problem of assigning elements of one set to elements of another set is motivated by important real-world scenarios like assigning students to universities, applicants to jobs and so on. In many of these applications, members of one or both the sets rank each other in an order of preference. The goal is to compute an assignment that is ``optimal" with respect to the preferences.

In this paper we focus on the {\em one-sided} preference list model where members of one set rank  a subset of  elements in the other set in a linear order (that is, preferences are assumed to be strict). Several notions of optimality like pareto-optimality, rank-maximality, fairness and popularity have been considered in literature (We give formal definitions of each of these notions later). For each of the above mentioned notions of optimality, there are efficient algorithms studied in the literature to compute the specified optimal matching. \cite{pareto} describe an algorithm that computes a maximum cardinality pareto-optimal matching. \cite{pop-main} present an algorithm to compute a popular matching while \cite{irving2004rank,fair} propose algorithms that optimize the head/tail of the matching profile (rank-maximal and fair respectively). Maximizing one metric could however result in poor performance on other yardsticks of measure. When comparing two matchings, it is difficult to measure the quality of the two matchings using a single scalar value. They can be compared using a variety of metrics like cardinality, number of matched Rank-1 edges or cardinality, none of which can serve as a sole indicator of optimality.

Profile based matchings, like Rank-maximal or Fair matchings which optimize for the head or the tail of the profile can turn out to be biased under certain circumstances. An alternative is to consider the \textit{Area under Profile Curve Ratio} metric introduced in \cite{experiment2}. This metric aims to maximize a measure, that is a weighted sum of matched edges, with the weight proportional to its position in the preference list

In this work, we first present a comprehensive experimental study of the well-studied notions of optimality and compare them using different measures of matching quality. We then describe the AUPCR metric, and propose algorithms to compute an AUPCR maximizing matching, and a maximum cardinality AUPCR maximizing matching.

Finally, we empirically evaluate different matching algorithms on synthetic graphs generated from generator models specified by \cite{experimental} using various metrics. The generated graphs fall into two categories, one having uniformly random preference lists and the other having highly correlated preference lists. Our analysis is inspired by the analysis of \cite{experimental}, and we additionally consider a ranking system in which the matchings are ranked based on multiple metrics. These rankings are consequently aggregated to obtain a single rank, which we use as a coarse indicator of matching quality.

The AUPCR maximizing matching is experimentally shown to have good performance across evaluated metrics on the considered data-sets, and we believe this matching is well suited for practical applications.

\section{Preliminaries}
Consider a set $\mathcal{A}$ of applicants and a set $\mathcal{P}$ of posts.
Every applicant $a$ has preference list over a subset of the posts in $\mathcal{P}$.
This list is a linear order (strict list) and is called the preference list of $a$ over $\mathcal{P}$.
The problem is readily represented as a bipartite graph with vertices $\mathcal{V} = \mathcal{A} \cup \mathcal{P}$ and
an edge $(a, p)$ is present if $p$ is acceptable to $a$.
Preferences of applicants are encoded by assigning ranks to edges.
Each edge $(a,p)$ has a rank $i$ if $a$ considers $p$ as its $i$-th most preferred post.
A matching $M \subseteq E$ is a collection edges such that no two edges share an endpoint.
Let $|A \cup \mathcal{P}| = n$ and $|E| = m$.
We now define formally the different notions of optimality.

% This work considers the task of finding a matching on this bipartite graph constructed from the preference lists. One seeks matchings that satisfy some notion of optimality, which we described subsequently.

\subsection{Maximum Cardinality Pareto Optimal Matching}
%\textit{Pareto Optimality} is a notion of optimality borrowed from long standing literature in economics. 
A matching $M$ is said to be Pareto-optimal if there is no other matching $M'$ such that some applicant is better off in $M'$ while no applicant is worse off in $M'$ than in $M$(an applicant is worse of in $M$ if it is matched to an less preferred vertex compared to $M'$) . Maximum cardinality Pareto optimal matchings(POM) can be computed in $\mathcal{O}(m\sqrt{n})$ time using the algorithm given by \cite{pareto}.

\subsection{Rank Maximal Matching}
The notion of \textit{rank-maximal} matchings was first introduced by
Irving under the name of \textit{greedy matchings} \cite{irving2003greedy}.
A rank-maximal matching is a matching in which the number of rank one edges is maximized, subject to which the number of rank two edges is maximized and so on.
Another way of defining rank-maximal matchings is through their $signatures$.
Given that $r$ is the largest rank given to a choice across all preference lists, we define the signature of a matching $M$ as an
 r-tuple $(x_{1},x_{2},\cdots ,x_{r}, x_{r+1})$ where, for $1 \leq i \leq r$, $x_{i}$ represents the number of applicants matched
to one of their $i$-th preferences ($x_{r+1}$ denotes the number of unmatched applicants). Let $(x_1, x_2,\cdots, x_r, x_{r+1})$ and $(y_1, y_2,\cdots, y_r, y_{r+1})$ denote
the signatures of $M$ and $M'$ respectively. We say that $M \succ$ $M'$ w.r.t. rank-maximality if there exists an index $k$ such that
 $1 \le k \le r$ and
for $1 \le i \le k$, $x_i = y_i$ and $x_k > y_k$. A matching $M$ is rank-maximal if $M$ has the best signature
w.r.t. rank-maximality. % there  A rank maximal matchings is one with lexicographically largest signature.
Through the rest of the paper, we denote this matching as RMM.
For the purposes of our experimental evaluation, we implement a simple combinatorial algorithm \cite{irving2004rank} to compute
a rank maximal matching. The running time of the algorithm is  $O(min(C\sqrt{n}, n+C)\cdot m)$. % time algorithm which also handles instances with ties. Here, $C$ is the maximum rank of any edge in the rank-maximal matching. 
\subsection{Popular Matching}
%The notion of popular matching was first introduced in the context of the stable matching problem \cite{pop-intro}. 
To define popularity, we translate preferences of applicants over posts to preferences of applicants over matchings.
An applicant $a$ prefers matching $M$ to $M'$ if either $a$ is matched in $M$ and unmatched in $M'$, or $a$ is matched in both $M$ and $M'$ but has better rank in $M$ than in $M'$. A matching $M$ is more popular than $M'$ if the number of applicants
who prefer $M$ to $M'$ is more than the number of those who prefer $M'$ to $M$. A matching $M$ is popular if there is no matching that is more popular than $M$.
A linear time  algorithm to compute a maximum cardinality popular matching for strict preferences is given in \cite{pop-main}.
The more popular than relation is not transitive, and hence it is possible that a popular matching does not exist.  When a popular matching does not exist, one can attempt to obtain the  least unpopular matching. We consider the unpopularity factor given in \cite{unpop}.
An algorithm given in \cite{pop-bounded} finds a popular matching if it exists. %, and an approximation to a minimum unpopular matching otherwise. This approximation ratio can however be as bad as $O(n)$ in the worst case. 
Through the rest of the paper, we denote this matching as POPM.

\subsection{Fair Matching}
Fair matchings can be considered as complementary to rank-maximal matchings. A fair matching is always a maximum cardinality matching, subject to this, it matches the least number of applicants their last preferred post, subject to this, least number of applicants to their second last
preferred post and so on. Fair matchings can be conveniently defined using signatures.
Let $(x_1, x_2,\cdots, x_r, x_{r+1})$ and $(y_1, y_2,\cdots, y_r, y_{r+1})$ denote
the signatures of two matchings  $M$ and $M'$ respectively. We say that $M \succ M'$ w.r.t. fairness
if there exists a index $k$, such that $1 \le k \le r+1$ and
for $k < i \le r+1$, $x_i = y_i$, and $x_k < y_k$.
A matching is fair if it is of maximum cardinality, and subject to that it has the best signature according to the above defined criteria.
Recently  \cite{fair} gave a combinatorial algorithm to compute fair matchings.
Through the rest of the paper, we denote this matching as FM.

\section{AUPCR Maximizing Matching}
Fair and Rank-maximal matchings are profile based matchings that are geared towards minimizing the tail or maximizing the head of the profile. However, optimizing for the peripheral portions of a profile may not be necessarily representative of a \textit{good} matching in many practical settings. This encouraged us to look into a metric called \textit{Area Under Profile Curve Ratio} (AUPCR) which, in a sense, seemed to capture the entire signature of a matching.
\subsection{Formulation of AUPCR}
The \textit{Area Under Profile Curve Ratio} (AUPCR), introduced under the context of matchings by \cite{experiment2} is a measure of second order stochastic dominance of the profile. It is a useful metric that can be used to compare multiple signatures and is very similar in nature to the highly popular Area Under Curve of Receiver Operating Characteristic \cite{roc}.

For a matching $M$ of a bipartite graph $G(A\cup \mathcal{P}, E)$ with $n_i(M)$ representing the number of applicants matched to their $i$'th preference, AUPCR(M) is defined as the ratio of \textit{Area Under Profile Curve} (AUPC) and \textit{Total Area} (TA) where
\begin{equation} 
\text{TA}(M) = {|A||\mathcal{P}|} 
\end{equation}
\begin{equation*} 
\text{AUPC}(M) = {\sum_{r=1}^{|\mathcal{P}|}|(a_{i},p_{j}) \in M:rank(a_{i},p_{j})\leq r|} 
\end{equation*}
\begin{equation} \label{defAupc}
= {\sum_{i=1}^{|\mathcal{P}|} (|\mathcal{P}|-i+1) n_i(M)}
\end{equation}
Giving us,
\begin{equation} \label{defAupcr}
\text{AUPCR}(M) = \frac{\sum_{i=1}^{|\mathcal{P}|} (|\mathcal{P}|-i+1) n_i(M)}{|A||\mathcal{P}|}
\end{equation}
One can visualize this quantity by considering Figure \ref{AUPCR}. For an instance with $|A|=8$, $|\mathcal{P}|=6$ and signature of matching $M$ given by $(4,0,2,1,1,0)$, the area under the shaded region corresponds to AUPC($M$) ($= 4+4+6+7+8+8 = 37$) while the area of bounding rectangle corresponds to TA($M$)($=8 \times 6=48$). With these computed, AUPCR($M$) is essentially the ratio of the two and is given by $\frac{37}{48} \approx 0.771$.
\begin{figure}[h]
	\centering
	\includegraphics[scale=0.35]{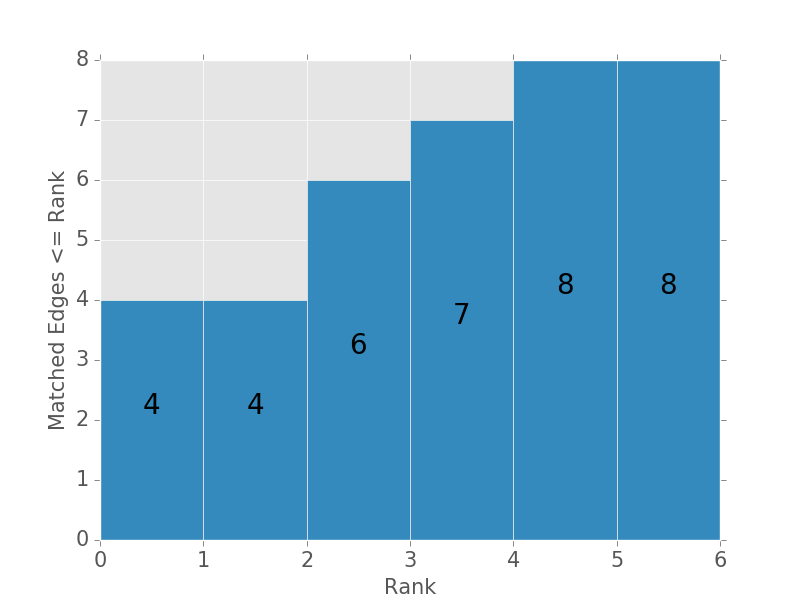}
    \caption{AUPCR - Visualization}
    \label{AUPCR}
\end{figure}

A matching that maximizes this measure can be vaguely seen as a "softer" version of the rank maximal matching: it does not give up matching low ranked edges entirely in order to match a large number of high ranked edges. Based on this we consider two problems:
\begin{itemize}
\item AUPCR Maximizing Matching - the problem of finding a matching which maximizes the AUPCR metric. We denote such a matching as AMM.
\item Max Cardinality AMM - the problem of finding a matching with the maximum cardinality among all matchings with maximum AUPCR. We denote such a matching as MC-AMM.
\end{itemize}

In this paper, we formulate algorithms to address the above defined problems and show that the Max Cardinality AUPCR maximizing matching performs favorably on a variety of other standard metrics typically used to compare matchings in practical settings. 

\subsection{Algorithm - AUPCR Maximizing Matching}
The problem of finding an AUPCR maximizing matching can be reduced to the problem of finding a maximum weighted perfect matching. Given a bipartite graph $G(A\cup \mathcal{P}, E)$ and a weight $w_e$ for each edge $e \in E$, we define the weight of a matching as $w(M) = \sum_{e \in E} w_e$. Then, the maximum weighted perfect matching problem is find a matching $M$ which matches all vertices in $A$ ($M$ is a perfect matching) and maximizes $w(M)$.

Given an bipartite graph $G(A\cup \mathcal{P}, E)$ with edges representing preferences of A, we construct $G'(A' \cup \mathcal{P}', E')$ as follows: \begin{enumerate}
\item $A' = A_1 \cup \mathcal{P}$ and $\mathcal{P}' = \mathcal{P}_1 \cup A_2$ where $A_1, A_2$ are copies of $A$ and $\mathcal{P}_1, \mathcal{P}_2$ are copies of $\mathcal{P}$.
\item For each edge $e \in E$ of rank $i$, add edge between corresponding vertices of $A_1$ and $\mathcal{P}_1$ with weight $|\mathcal{P}| - i+1$. Similarly, add an edge between $A_2$ and $\mathcal{P}_2$ with the same weight.
\item Add edges with weight 0 from vertices in $A_1$ to their copies in $A_2$. Add similar edges between $\mathcal{P}_1$ and $\mathcal{P}_2$. We refer to these edges as identity edges.
\end{enumerate}
\subsubsection{Proof of Correctness} ~\\ 
\begin{claim}
If $M$ is a max weighted perfect matching in $G'$, then $M$ restricted to $A_1 \cup \mathcal{P}_1$ is a AMM in $G$.
\end{claim}
\begin{proof}
Let $M_1$ be the matching obtained by restricting $M$ to $A_1 \cup \mathcal{P}_1$ and $M_2$ obtained by restricting $M$ to $A_2 \cup \mathcal{P}_2$. Since $M$ is a perfect matching, all vertices of $G'$ must be matched. So, if a vertex in $A_1 \cup \mathcal{P}_1$ is not matched in $M_1$, it must be matched to its copy in $A_2 \cup \mathcal{P}_2$ via the identity edge. This means that its copy is also unmatched in $M_2$. So, $M_1$ and $M_2$ match the same set of vertices.  Since the identity edges have 0 weight, 
$$w(M) = w(M_1) + w(M_2)$$

Since $M_1$ and $M_2$ match the same set of vertices, one can copy the edges matched in $M_1$ to $A_2 \cup \mathcal{P}_2$. This means that $w(M_1) = w(M_2)$ and $w(M) = 2w(M_1)$. Maximizing $w(M)$ is equivalent to maximizing $w(M_1)$. 

We also have 
\begin{align*}\label{eq1}
  w(M_1) & = \sum_{e \in M_1} w_e = \sum_{e \in M_1} |\mathcal{P}|-r_e+1 \\
  & = \sum_{i=1}^{|\mathcal{P}|} n_i (|\mathcal{P}|-i+1) = |A||\mathcal{P}|\text{AUPCR}(M_1) 
\end{align*}
where $r_e$ is the rank of edge $e$ and $n_i$ is the number of edges of $M_1$ with rank $i$.This means that maximizing $w(M_1)$ maximizes $\text{AUPCR}(M_1)$. 

Hence, if $M$ is a maximum weight perfect matching in $G'$, $M_1$ is a max AUPCR matching in $G$.
\end{proof}

\subsection{Algorithm - Max Cardinality AMM}
The problem of finding a Max Cardinality AMM can also be reduced to an instance of max weighted perfect matching. The reduction is the same as the max AUPCR case, but we add a negative weight of $-\frac{1}{|A|+|\mathcal{P}|}$ to the identity edges going from $A_1$ to $A_2$.
\subsubsection{Proof of Correctness} ~\\ 
\begin{claim}
If $M$ is a max weighted perfect matching in $G'$, then $M$ restricted to $A_1 \cup \mathcal{P}_1$ is a Max Cardinality AMM in $G$.
\end{claim}
\begin{proof}
As before, we can prove that $w(M_1) = w(M_2)$. However, $w(M) = w(M_1) + w(M_2) + w(I)$ where $I$ is the set of identity edges from $A_1$ to $A_2$ in $M$. If $M_1$ leaves $k_A$ vertices in $A$ unmatched and $k_{\mathcal{P}}$ vertices in $\mathcal{P}$ unmatched, then $M_2$ also leaves the same vertices unmatched. So, we have $2(k_A + k_{\mathcal{P}})$ identity edges in $I$ and hence $$w(I) = -2\frac{k_A+k_{\mathcal{P}}}{|A|+|\mathcal{P}|}$$ 
Since $k_A + k_{\mathcal{P}} < |A| + |\mathcal{P}|$, we have $-2 < w(I) \le 0$ and
$$2(w(M_1)-1) < w(M) \le 2w(M_1)$$
Let $M'$ be a Max AUPCR matching extended to $G'$ and $M_1'$ be its restriction to $A_1 \cup \mathcal{P}_1$. Since $w(M') \le w(M)$
\begin{align*}
	& 2(w(M_1')-1) < 2w(M_1) \\
    \Rightarrow & w(M_1') - w(M_1) < 1 \\
    \Rightarrow & |A||\mathcal{P}|\text{AUPCR}(M_1') - |A||\mathcal{P}|\text{AUPCR}(M_1) < 1\\
    \Rightarrow & \text{AUPCR}(M_1') - \text{AUPCR}(M_1) < \frac{1}{|A||\mathcal{P}|}
\end{align*}
From the definition of AUPCR, we can see that if two matchings have different AUPCR, then the difference is $\ge \frac{1}{|A| |\mathcal{P}|}$. So, $M_1'$ and $M_1$ have the same AUPCR, which means that $M_1$ is an AUPCR maximizing matching in $G$.

The cardinality of $M_1$ is $|M_1| = |A|-k_A = |\mathcal{P}|-k_{\mathcal{P}}$. Writing $w(M)$ in terms of $|M_1|$, 
$$w(M) = 2w(M_1) - \frac{|A|+|\mathcal{P}|-2|M_1|}{|A|+|\mathcal{P}|}$$
  All AUPCR maximizing matchings will have the same $w(M_1)$, which means that maximizing $w(M)$ maximizes $|M_1|$. So, $M_1$ is a maximum cardinality AUPCR maximizing matching in $G$.
\end{proof}

The time complexity of the algorithm to find maximum weighted matching presented is $O(m\sqrt{n}\log n)$ \cite{maxweight}. Since both our algorithms construct a graph with $2n$ vertices and $m + n^2$ edges and find a max weighted matching, the time complexity would be $O(n^2\sqrt{n}\log n)$.

% \begin{theorem}
% Consider a matching problem where all applicants have $\le 2$ preferences, take two  matchings $M_1$ and $M_2$. Cardinality of $M_1 >$ Cardinality of $M_2 \Rightarrow $ AUPCR of $M_1 \ge$ AUPCR of $M_2$.
% \end{theorem}
% \begin{proof}
% Signature of $M_1 = (a_1,b_1)$ and signature of $M_2 = (a_2, b_2)$. Let $n$ = number of posts. Then, $a_1+b_1 \le n, a_1 \le n$.

% Cardinality of $M_1 = a_1 + b_1$, AUPCR of $M_1 = na_1 + (n-1)b_1$. Assume $a_1 + b_1 > a_2 + b_2 \Rightarrow a_1 + b_1 - (a_2 + b_2) = k > 0$. 

% AUPCR($M_1$) - AUPCR($M_2$) $ = na_1 + (n-1)b_1 - na_2 + (n-1)b_2 = (n-1)k + a_1 - a_2$
% AUPCR($M_1$) $<$ AUPCR($M_2$) $\Rightarrow (n-1)k + a_1 - a_2 < 0 \Rightarrow (n-1)k \le a_2 - a_1 \le n \Rightarrow k = 1$. $(n-1) < a_2 - a_1 \Rightarrow $ only possibility $a_1 = 0, a_2 = n \Rightarrow a_2 + b_2 \ge n \ge a_1 + b_1$

% \end{proof}

\begin{table}[h]
\small
\begin{center}
\begin{tabular}{lll}
\multicolumn{1}{c}{\bf Algorithm}  
&\multicolumn{1}{c}{\bf Running Time}
\\ \hline  \\
\textbf{POM}  \cite{pareto} & $\mathcal{O}(m\sqrt{n})$\\
\textbf{RMM} \cite{irving2004rank} & $\mathcal{O}(min(C\sqrt{n}, n+C)m)$\\
\textbf{FM}  \cite{fair} & ${\mathcal{O}}(Cm\sqrt[]{n}\log n)$\\
\textbf{POPM}  \cite{pop-bounded} & $\mathcal{O}(m\sqrt{n})$\\
\textbf{AMM, MC-AMM} (Our work) & ${\mathcal{O}}(n^2\sqrt{n}\log n)$\\
\end{tabular}
\end{center}
\caption{$C$ is the maximum rank of any edge in the matching, $n=|A|+|P|$ and $m=|E|$}
\label{tab:algorithms}
\end{table}

\section{Experiments}
\subsection{Evaluation Metrics}
The matchings obtained from each algorithm are evaluated with respect to the following metrics.
\begin{itemize}
\item \textbf{Cardinality}: The number of edges present in the matching.
\item \textbf{Unpopularity measure}: The unpopularity measure $u(M, M')$ measures how far away a matching $M$ is from a popular(least unpopular) matching $M'$. Let $p(M_1,M_2)$ be the number of applicants that prefer $M_1$ over $M_2$.  Then $u(M,M')$ for matching $M$ is defined as the ratio of $p(M',M) - p(M,M')$ to the total number of applicants. 
\item \textbf{Rank 1}: The number of matched \textit{rank 1} edges
\item \textbf{AUPCR}: The AUPCR metric is second order stochastic dominance of the profile as defined in Equation \ref{defAupcr}.
\item \textbf{Ranks less than half the preference list size (RHPL)}: This counts the number of applicants who have been matched to a post with a rank better than or equal to half the length of their preference list.
\item \textbf{Average rank}: For a matching $M$, this is the average rank of all matched edges. Although this is similar to the AUPCR metric, the average rank is computed only over the matched edges while AUPCR accounts for unmatched edges.
\item \textbf{Worst rank}: For a matching $M$, this is the highest (worst) rank among all matched edges in $M$.
\item \textbf{Time}: The time taken to find the matching. This is implementation dependent, and the algorithms used have been mentioned earlier along with their time complexities.
\end{itemize}

\subsection{Instances}
For our experiments, we consider two structured instance generators, namely: \textit{Highly Correlated} and \textit{Uniform Random}. These generators are similar in nature to \cite{experimental}, but we consider only instances with strict preference lists. Though all the algorithms described above, except Maximum Cardinality Pareto Optimal, can also handle instances with ties, we went with this choice to have a set of instances upon which all the algorithms could be compared and analyzed. If one thinks about it, this choice is not too restrictive as in practical scenarios preference lists are often strict and devoid of ties.

\subsubsection{Uniform Random (UNI)}
Similar to HC, UNI instances are also parameterized by a density $d$ with $0 \le d \le 1$. Every applicant has a preference list size of $l = \left\lfloor n_{p}\cdot d\right\rfloor$. These preference lists are chosen uniformly at random from the set of permutations of $l$ posts. Let an applicant $a$'s adjacency list be $(p_{1}, p_{2}, ..., p_{l})$. Then $p_{1}$ is ranked 1 by $a$, $p_{2}$ is ranked 2, and so on. Unlike HC, preference list length is identical across all applicants.

\subsubsection{Highly Correlated (HC)}
These instances are generated based on a global preference ordering(say $P$) for the set of posts; one that all the applicants agree upon. A HC instance is parameterized by a density $d$ with $0 \le d \le 1$. For every vertex pair $(u, v)$ with $u \in A, v \in \mathcal{P}$, an edge $(u, v)$ is added with a probability $d$. Once the graph has been constructed, the applicants rank the posts as per the global preference list: the best post, as per P, an applicant is connected to is assigned rank 1, and so on. 

\begin{figure*}[t]
	\centering
    \subfloat[AUPCR]{
 		\includegraphics[scale=0.22]{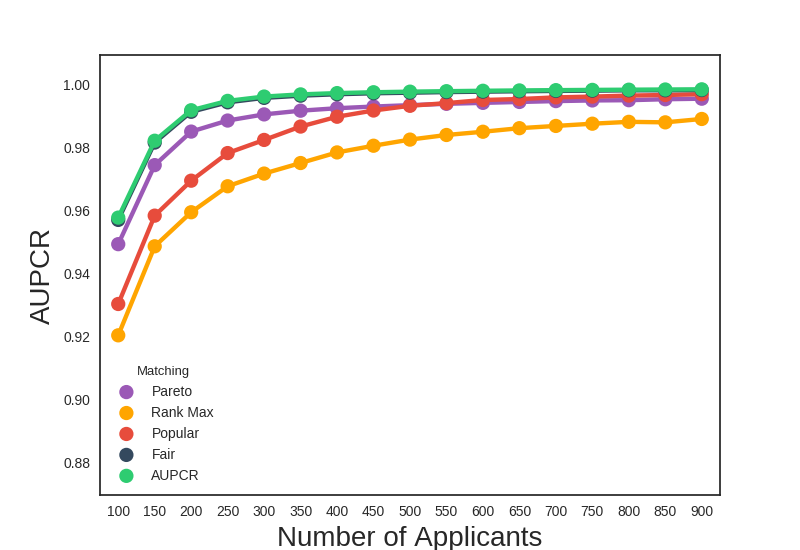}
	}
    \subfloat[Avg Rank]{
 		\includegraphics[scale=0.22]{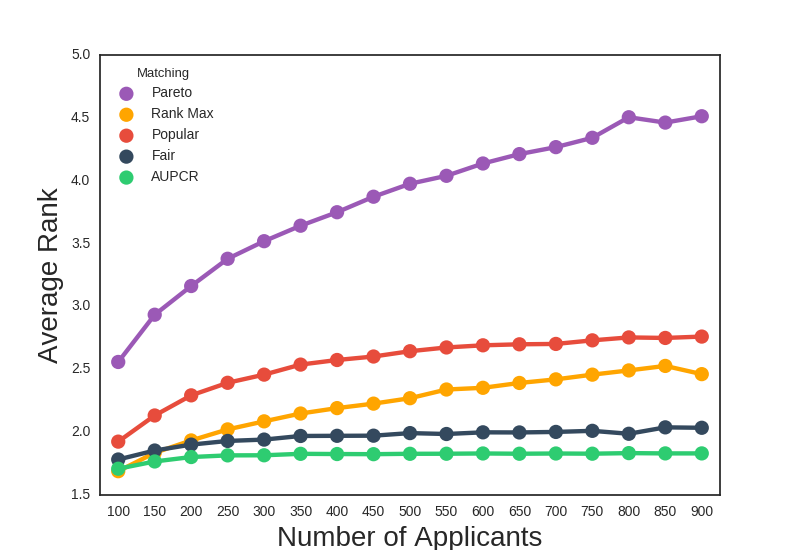}
	}
    \subfloat[Popularity]{
 		\includegraphics[scale=0.22]{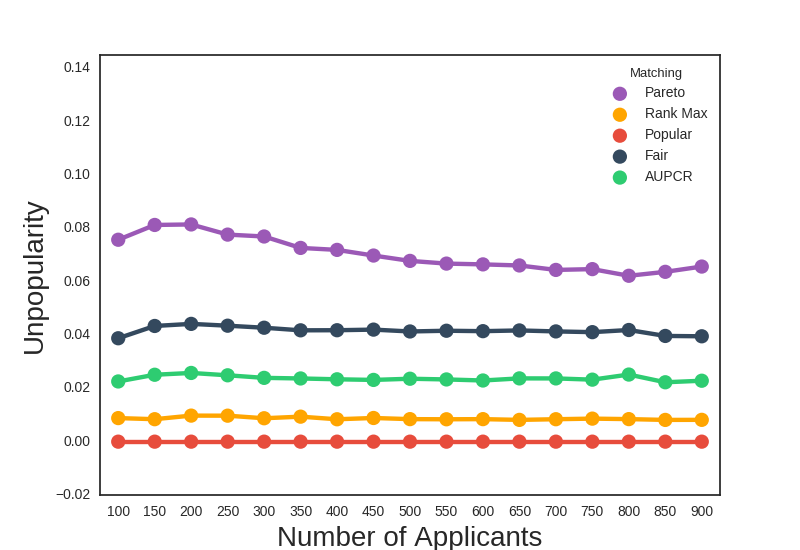}
	}
	\subfloat[Rank 1]{
  		\includegraphics[scale=0.22]{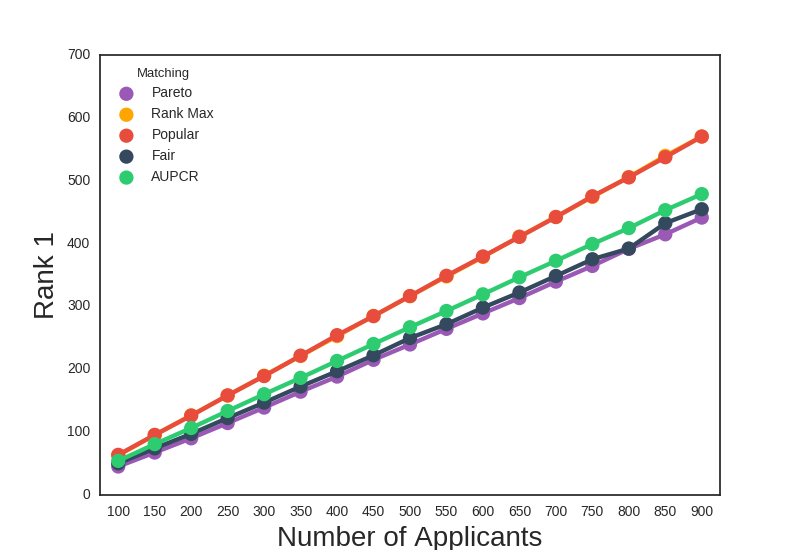}
	}
    \caption{Uniform Random}
\end{figure*}

\begin{figure*}[t]
	\centering
    \subfloat[AUPCR]{
 		\includegraphics[scale=0.22]{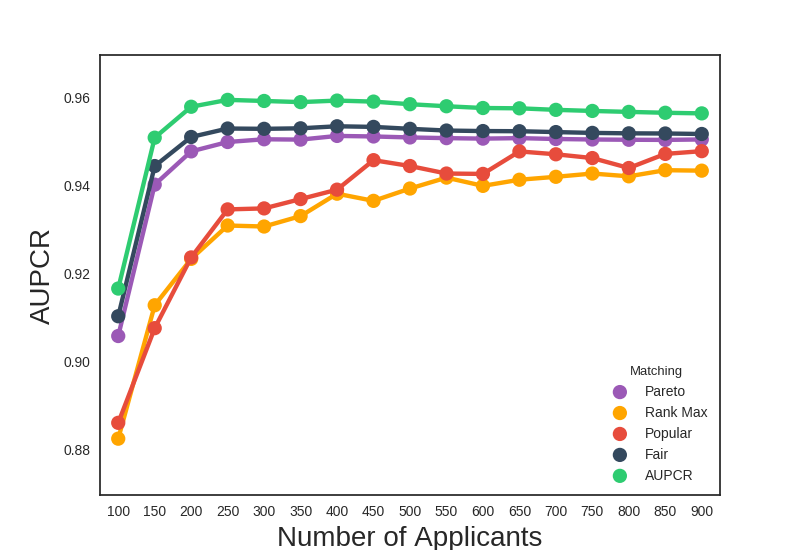}
	}
    \subfloat[Avg Rank]{
 		\includegraphics[scale=0.22]{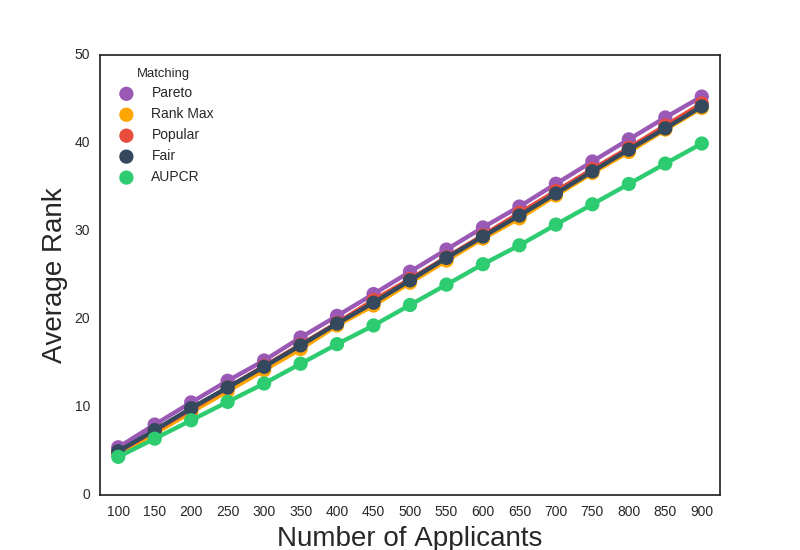}
	}
    \subfloat[Popularity]{
 		\includegraphics[scale=0.22]{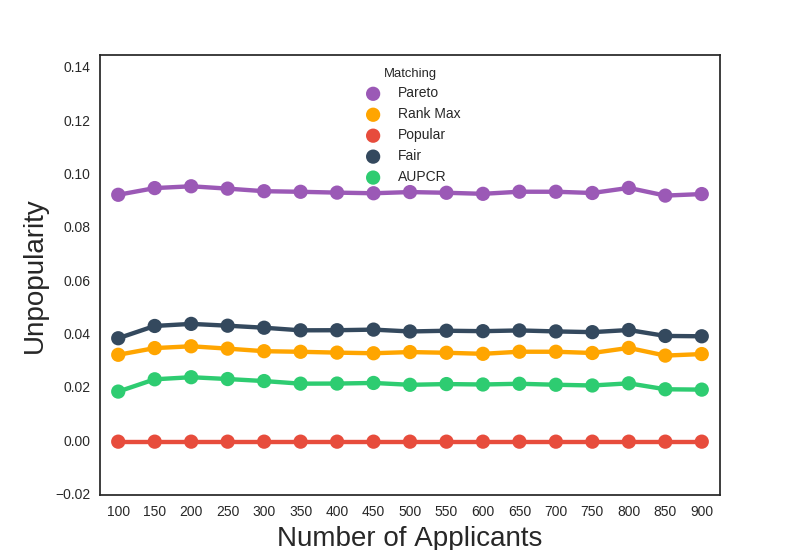}
	}
	\subfloat[Rank 1]{
  		\includegraphics[scale=0.22]{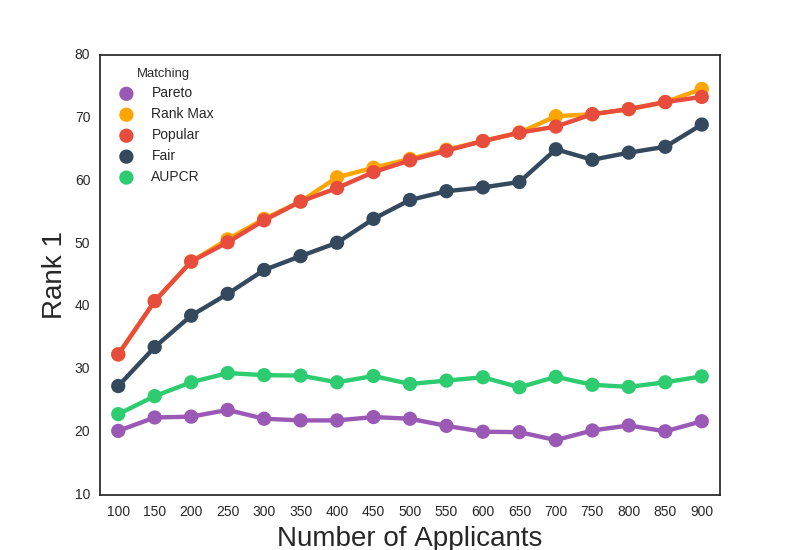}
	}
    \caption{Highly Correlated}
\end{figure*}

\subsection{Experiment Setup}
The number of applicants are equal to the number of posts in any graph and is varied from $50$ to $900$ in steps of $50$. Orthogonally, the density parameter $d$ for HC and UNI is varied from 0.02 to 0.20 in steps 0.02. The reason for this choice of range is that real world datasets are not very dense in nature . Each instance is averaged over 50 random seeds. There exists one more level of averaging across different density($d$) values to get one value for each metric for each problem size(number of applicants).

The variant of Max-AUPCR, that does not not enforce maximum cardinality is used. Surprisingly, this still yields a max-cardinality matching without exception. For POPM, in cases where popular matchings don't exist, the least unpopular matching is utilized. The code was executed using the Amazon web services(AWS) based EC2 service on a t2.micro instance(1 GB Ram, 1 CPU, Intel Xeon processor).
\subsection{Experimental Results}
\subsubsection{Comparing matchings based on rank means} ~\\
\begin{table}[h]
\setlength\tabcolsep{4.2pt}
\begin{center}
\begin{tabular}{llllll}
\multicolumn{1}{c}{}  
&\multicolumn{1}{c}{\bf POM}
&\multicolumn{1}{c}{\bf RMM}
&\multicolumn{1}{c}{\bf POPM}
&\multicolumn{1}{c}{\bf FM}
&\multicolumn{1}{c}{\bf AMM}
\\ \hline  \\
\textbf{Card.} & \textbf{1.00} & 2.94 & 2.00 & \textbf{1.00} & \textbf{1.00} \\
\textbf{Unpop.} & 5.00 & 2.00 & \textbf{1.00} & 3.99 & 3.01 \\
\textbf{Rank 1} & 3.80 & \textbf{1.00} & \textbf{1.00} & 3.18 & 2.00 \\
\textbf{AUPCR} & 3.81 & 4.85 & 3.28 & 1.99 & \textbf{1.00} \\
\textbf{RHPL} & 4.28 & 2.89 & 2.58 & 1.26 & \textbf{1.14} \\
\textbf{Avg Rank} & 4.98 & 2.58 & 3.81 & 2.34 & \textbf{1.22} \\
\textbf{Worst Rank} & 4.66 & 3.27 & 3.39 & \textbf{1.00} & 1.94 \\

\\ \hline \\
\textbf{Rank Mean} & 3.93 & 2.79 & 2.44 & 2.11 & \textbf{1.62} \\
\end{tabular}
\end{center}
\caption{Rank means of algorithms on metrics for UNI}
\label{tab:rank_metrics}
\end{table}

\begin{table}[h]
\setlength\tabcolsep{4.2pt}
\begin{center}
\begin{tabular}{llllll}
\multicolumn{1}{c}{}  
&\multicolumn{1}{c}{\bf POM}
&\multicolumn{1}{c}{\bf RMM}
&\multicolumn{1}{c}{\bf POPM}
&\multicolumn{1}{c}{\bf FM}
&\multicolumn{1}{c}{\bf AMM}
\\ \hline  \\
\textbf{Card.} & \textbf{1.00} & 2.89 & 2.11 & \textbf{1.00} & \textbf{1.00} \\
\textbf{Unpop.} &  5.00 & \textbf{3.07} & \textbf{1.00} & 3.94 & 2.00 \\
\textbf{Rank 1} & 4.93 & \textbf{1.00} & 1.59 & 2.99 & 3.96 \\
\textbf{AUPCR} & 3.25 & 4.84 & 3.92 & 2.00 & \textbf{1.00} \\
\textbf{RHPL} &  4.09 & 1.96 & 3.02 & 4.75 & \textbf{1.11} \\
\textbf{Avg Rank} &  5.00 & 2.18 & 3.53 & 3.20 & \textbf{1.09} \\
\textbf{Worst Rank} & 3.66 & 4.73 & 3.42 & 1.01 & 1.99 \\

\\ \hline \\
\textbf{Rank Mean} & 3.84 & 2.95 & 2.66 & 2.69 & \textbf{1.73} \\
\end{tabular}
\end{center}
\caption{Rank means of algorithms on metrics for HC}
\label{tab:rank_metrics2}
\end{table}
For this analysis, we consider a set of the evaluation metrics which we believe characterizes preference matchings in general. For a given metric and graph instance, we rank the algorithms in terms of performance with the best one getting a rank of 1 and worst one getting a rank of 5. We then average this rank across all instances and this value corresponds to an entry in Table \ref{tab:rank_metrics}. The \textit{rank mean} is computed by taking the average of the entries along the column. This value is intended to serve as a measure of overall performance.

As seen from the table, each chosen metric has a subset of the algorithms performing best. It is however important to note that AMM performs competitively in almost all metrics. This observation is also qualitatively supported from the fact that the \textit{rank mean} attained by AMM is lowest among all algorithms for both UNI and HC instances. This empirically shows that AMM is able to achieve a much desired balance, making it a very compelling choice for many practical preference matching problems.

\subsection{Comparing the Matchings on different metrics}
Some interesting observations for some metrics are as follows :
\begin{itemize}
\item \textbf{Cardinality} : As expected, POM and FM have the largest cardinality since they compute maximum cardinality matchings. However, it was observed that AMM without exception returned a maximum cardinality matching. While this may not universally true(as proved in consequent section) this is a useful property in practice
\item \textbf{RHPL} : The RHPL is one metric that no matching in particular optimizes for. It is peculiar to note that  AUPCR maximizes this metric indicating that it is indeed a more general notion of optimality.
\item \textbf{Rank 1} : It was observed that both popular and rank maximal matchings have similar if not same number of rank 1 edges. While the head of the signature is maximized, it is observed that both these matchings display poor performances on metrics that account for the entirety or the tail of the signature.
\item \textbf{Time} : Dictated by the computational time complexities of the respective algorithms, the times were vastly different for FM and AMM compared to the other three matchings. In graphs with 900 vertices(in each partition), the FM took 512.45 seconds,AMM executed in 204.78 seconds while POP and PM were executed in less than 5 seconds.  
\end{itemize}

\subsection{Understanding AMM}
The strongly positive empirical performance of AMM, in various metrics of importance as shown above, leads us to ask some interesting questions.

\subsubsection{Is an AMM Pareto optimal?} ~ \\
Yes, AMM is a Pareto optimal matching.
\begin{theorem}
AUPCR maximizing matching is Pareto optimal.
\end{theorem}
\begin{proof}
Assume to the contrary that an AUPCR maximizing matching $M$ is not Pareto optimal. This means there exists a matching $M'$ where every applicant in $M'$ is at least as well off as in $M$ and at least one applicant in $M'$ is better off than $M$. Consider a vertex $v \in A$. Let $r_{M}(v)$ be the rank of the post that $v$ is matched to ($r_{M}(v) = |\mathcal{P}|+1$ if $v$ is unmatched), and $r_{M'}(v)$ be defined analogously.
\begin{align*}
&AUPCR(M') - AUPCR(M) \\
&= \sum_{v \in A}((\mathcal{P} - r_{M'}(v) + 1) - (\mathcal{P} - r_{M}(v) + 1)) \\
&= \sum_{v \in A}(r_{M'}(v) - r_{M}(v)) \\
&> 0
\end{align*}
The last inequality follows from the fact that every term of the summation is non negative and at least one term is positive by our assumption that $M$ is not Pareto optimal.\newline
Since AUPCR($M'$) - AUPCR($M$) $> 0$, $M$ is not an AUCPR maximizing matching, a contradiction, and so $M$ must be Pareto optimal.
\end{proof}
\subsubsection{Is an AMM always a maximum cardinality matching?} ~ \\
An AMM need not always be a maximum cardinality matching. Consider the instance with $A$  = $\{a_1, a_2, a_3, a_4\}$, $P$  = $\{b_1, b_2, b_3, b_4\}$ and the preferences given by
\begin{align*}
&a_1 : (b_1, 1) \\
&a_2 : (b_1,1), (b_2,2) \\
&a_3 : (b_2,1), (b_1,2), (b_3,3) \\
&a_4 : (b_3,1), (b_1,2), (b_4,3) 
\end{align*}
As shown in Figure \ref{AMM_1} and Figure \ref{MM_1}, for this instance, AMM has a cardinality of 3 while a maximum cardinality matching has a cardinality of 4.
\begin{figure}[h]
	\centering
	\includegraphics[scale=0.35]{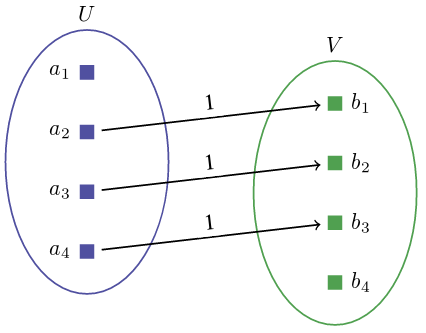}
    \caption{AUPCR Maximizing Matching}
    \label{AMM_1}
\end{figure}
\begin{figure}[h]
	\centering
	\includegraphics[scale=0.35]{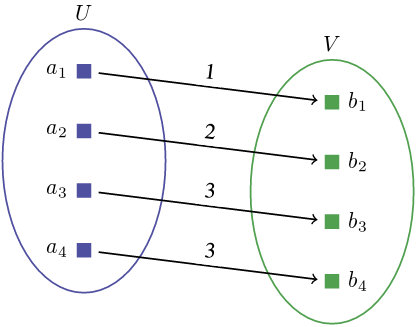}
    \caption{Maximum Cardinality Matching}
    \label{MM_1}
\end{figure}
\subsubsection{Do all AMMs have the same cardinality?}  ~ \\
All AMMs need not have the same cardinality. Consider the instance with $A$  = $\{a_1, ... a_6\}$ and $P$  = $\{b_1, ... b_6\}$  and the preferences given by
\begin{align*}
&a_1 : (b_6, 1), (b_3, 2), (b_1,3) \\
&a_2 : (b_2,1), (b_3,2), (b_1,3) \\
&a_3 : (b_4,1), (b_5,2), (b_2,3) \\
&a_4 : (b_1,1), (b_4,2), (b_6,3) \\
&a_5 : (b_5,1), (b_2,2), (b_1,3) \\
&a_6 : (b_4,1), (b_2,2), (b_5,3)
\end{align*}
As shown in Figure \ref{AMM_2} and Figure \ref{AMM_3}, both are AUPCR maximizing matchings, with an AUPCR of 0.833, but they have different cardinalities. This example also shows that multiple AMMs can exist for a given instance.
\begin{figure}[h]
	\centering
	\includegraphics[scale=0.35]{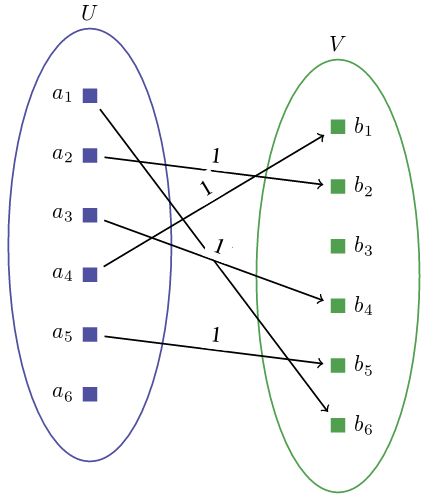}
    \caption{An AMM with $|M|=5$}
    \label{AMM_2}
\end{figure}
\begin{figure}[h]
	\centering
	\includegraphics[scale=0.35]{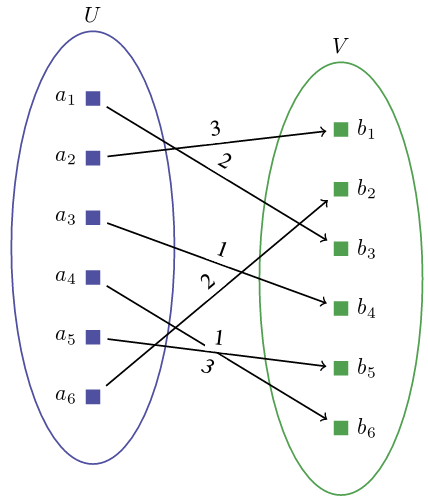}
    \caption{An AMM with $|M|=6$}
    \label{AMM_3}
\end{figure}
\subsubsection{Is an AMM always more "rank maximal" than a FM?} ~
An AMM matching need not be more rank-maximal than the fair matching.  Consider the instance with $A$  = $\{a_1, a_2, a_3, a_4, a_5, a_6, a_7\}$, $P$  = $\{b_1, b_2, b_3, b_4, b_5, b_6, b_7\}$ and the preferences given by
\begin{align*}
&a_1 : (b_1, 1) \\
&a_2 : (b_2,1) \\
&a_3 : (b_3,1), (b_4,2), \\
&a_4 : (b_1,1), (b_5,2), (b_4,3)  \\
&a_5 : (b_1,1), (b_6,2), (b_2,3), (b_5,4) \\
&a_6 : (b_1,1), (b_2,2), (b_7,3), (b_6, 4), (b_3, 5) \\
&a_7 : (b_7,1) 
\end{align*}
An AMM matching for the above graph is as show in Figure \ref{AMM_5} and it's signature is given by $(3,3,0,0,1)$. One can easily see that the FM shown in Figure \ref{AMM_6} is more rank-maximal with a signature $(4, 0, 1, 2, 0)$.

\begin{figure}[h]
	\centering
	\includegraphics[scale=0.35]{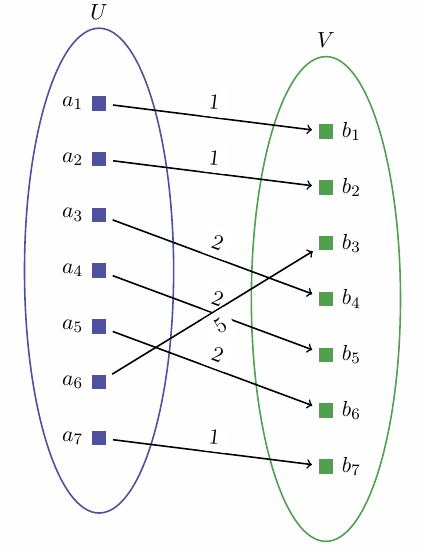}
    \caption{An AMM with matching with signature $(3,3,0,0,1)$}
    \label{AMM_5}
\end{figure}

\begin{figure}[h!]
	\centering
	\includegraphics[scale=0.35]{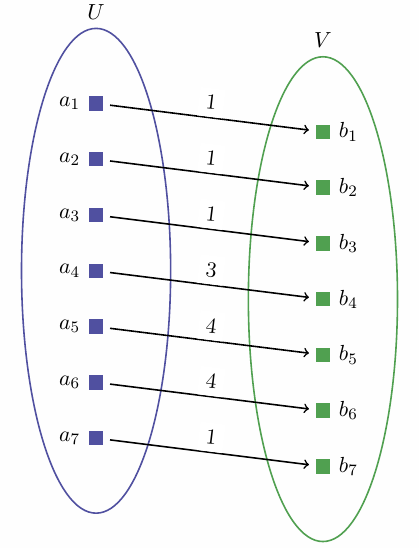}
    \caption{A Fair matching with signature $(4,0,1,2,0)$}
    \label{AMM_6}
\end{figure}

% \subsubsection{Is an AMM always more "fair" than a RMM?} ~ \\
% This question is again answered in the negative with a number of counter-examples observed on graphs with a slightly higher number of vertices( $n \sim 20$). The existence of these counter examples indicates that there is no guarantee on the nature of AMM matchings.

\section{Conclusion}
In this work, we introduce the notion of an AUPCR maximizing matching. We describe two variants with one maximizing the AUPCR, and the other maximizing the cardinality subject to maximizing the AUPCR. We empirically evaluate our algorithm on standard synthetically generated datasets and highlight that AUPCR maximizing matching achieves this much needed middle-ground with respect to the different notions of optimality. The overall performance of the AUPCR matching is superior in comparison to other matchings when all metrics are cumulatively used for comparison. Extending the AUPCR matching and finding algorithms with reduced time complexity is left as future work.

\newpage

\bibliography{main}
\bibliographystyle{aaai}
\end{document}